\documentclass[11pt]{llncs}

\usepackage{fullpage, amsmath, amssymb, subfigure, enumerate, color, times, graphicx}

\newcommand{\bigs}{{\mathcal S}}

\newcommand{\opt}{\text{OPT}}
\renewcommand{\P}{\text{P}}
\newcommand{\NP}{\text{NP}}
\newcommand{\DTIME}{\text{DTIME}}
\newcommand{\ZPTIME}{\text{ZPTIME}}

\renewcommand{\phi}{\varphi}
\renewcommand{\epsilon}{\varepsilon}

\title{Set Covering Problems with General Objective Functions}
\author{Jean Cardinal, Christophe Dumeunier}
\institute{Universit\'e Libre de Bruxelles (ULB)\\ 
Computer Science Department, CP 212\\
B-1050 Brussels, Belgium\\
{\tt \{jcardin,cdumeuni\}@ulb.ac.be}}
\begin{document}
\maketitle

\sloppy

\begin{abstract}
We introduce a parameterized version of set cover that generalizes several previously studied problems. 
Given a ground set $V$ and a collection of subsets $S_i$ of $V$, a feasible solution is a partition of $V$ such that each subset of the partition is included in one of the $S_i$. The problem involves maximizing the mean subset size of the partition, where the mean is the generalized mean of parameter $p$, taken over the elements. For $p=-1$, the problem is equivalent to the classical minimum set cover problem. For $p=0$, it is equivalent to the minimum entropy set cover problem, introduced by Halperin and Karp. For $p=1$, the problem includes the maximum-edge clique partition problem as a special case. We prove that the greedy algorithm simultaneously approximates the problem within a factor of $(p+1)^{\frac 1p}$ for any $p\in{\mathbb R}^+$, and that this is the best possible unless $\P = \NP$. These results both generalize and simplify previous results for special cases. We also consider the corresponding graph coloring problem, and prove several tractability and inapproximability results. Finally, we consider a further generalization of the set cover problem in which we aim at minimizing the sum of some concave function of the part sizes. As an application, we derive an approximation ratio for a Rent-or-Buy set cover problem.
\end{abstract}

\section{Introduction}

The greedy strategy is one of the simplest and most well-known heuristic, which can be applied to many combinatorial optimization problems. In the case of the minimum set cover problem, it involves iteratively choosing a subset that covers a maximum number of uncovered elements. We study this algorithm on a natural family of set covering problems in which the value of a subset depends on the number of elements it covers, and a parameter $p$ encodes the way in which these values are combined. This parameter interpolates between different versions of the set covering problem, in particular between the classical minimum set cover problem, the minimum entropy set cover problem, and the simpler problem of finding a subset of maximum size. 

Intuitively, the greedy algorithm should perform better for objective functions in which more importance is given to subsets covering many elements. We give a formal support to this intuition by showing that the greedy algorithm provides a constant factor approximation for all positive values of the parameter $p$. We further show that this is the best we can achieve unless $\P =\NP$.\\ 

We first define some notations. Let $V$ be an $n$-element ground set and $\bigs  =
\{S_1,\dots,S_k\}$ a collection of $k$ subsets of $V$, whose union is $V$. 
In the minimum set cover problem, we seek a minimum size 
subset $\mathcal{T}\subseteq\bigs $ such that $\bigcup_{S_i\in\mathcal{T}} S_i = V$.  
We define a {\em cover} as an assignment $\phi : V \mapsto \bigs $ 
of each element of $V$ to a set of $\bigs $ such that $v \in \phi (v)$ 
for all $v \in V$. This definition allows us to define alternative objective functions for 
the set cover problem. Given a cover $\phi$, let us define
a {\em part} as a set $\phi^{-1}(S_i)$ for some $S_i\in\bigs$. 
We use the following two notations: $c_i:=|\phi^{-1}(S_i)|$ is the {\em part size} of the $i$th subset $S_i$
with respect to $\phi$, and $a_v:=|\phi^{-1}(\phi (v))|$ is the size of the part containing the element $v$, with $v\in V$.\\

We define a new family of set cover problems in which we aim at maximizing the {\em mean}
$M (\{ a_v : v\in V\} )$ of the values $a_v$.
There exist many definitions of the mean $M (\{ a_1,a_2,\ldots ,a_n\} )$ of a set of numbers. 
The most widely used definition is the {\em arithmetic mean}: $M_1 (\{ a_1,a_2,\ldots ,a_n\} ) := \frac 1n \sum_{i=1}^n a_i$.
Another well-known definition is the {\em geometric mean}:
$M_0 (\{ a_1,a_2,\ldots ,a_n\} ) := (a_1\cdot a_2\cdot \ldots\cdot  a_n)^{\frac1n}$.
Finally, we also consider the {\em harmonic mean}:
$M_{-1} (\{ a_1,a_2,\ldots ,a_n\} ) := n /\left( \sum_{i=1}^{n} a_i^{-1} \right)$.
The arithmetic, geometric, and harmonic means are special cases of the {\em generalized} mean:
\begin{equation}
\label{eq:genmean}
M_p (\{ a_1,a_2,\ldots ,a_n\} ) = \left( \frac 1n \sum_{v\in V} a_v^p \right)^{\frac 1p} = \left( \frac 1n \sum_{i:c_i\not= 0} c_i^{p+1} \right)^{\frac 1p}.
\end{equation}
This value is the arithmetic mean for $p=1$, and the harmonic mean for $p=-1$. It is well-known that the limit of the generalized mean for 
$p\to 0$ is equal to the geometric mean. The generalized mean with parameter $p$ is also called the normalized 
$L_p$-norm\footnote{We use the word $p$-mean here, in order to avoid confusion with the ``minimum $L_p$-norm set cover" problem~\cite{GGKT07}.}.

\begin{definition}[Maximum $p$-mean set cover]
Given an $n$-element ground set $V$ and a collection $\bigs  =\{S_1,\dots,S_k\}$ 
of subsets of $V$ whose union is $V$, find a cover $\phi: V \mapsto \bigs $ that maximizes $M_p (\{ a_v : v\in V\} )$,
where $a_v:=|\phi^{-1}(\phi (v))|$, and $M_p$ is the generalized mean of parameter $p$.
\end{definition}

\subsection*{Special Cases}

Interestingly, letting $p=-1$ (harmonic mean) or $p=0$ (geometric mean) yields set cover problems that are already known: the harmonic mean version is the minimum set cover problem, while the geometric mean version is the {\em minimum entropy set cover} problem~\cite{CFJ06}. A special case of the maximum $p$-mean set cover problem for $p=1$ has recently been introduced in the form of a graph coloring problem~\cite{DJLLP07}.

\paragraph{Minimum Set Cover.} 
The maximum harmonic mean set cover problem can be cast as $\min_{\phi} \sum_{v\in V} \frac1{a_v}$. We can rewrite this objective function as
$\sum_{v\in V} \frac1{a_v} =  \sum_{S_i\in\bigs } \sum_{v\in \phi^{-1}(S_i)} \frac 1{c_i} = |\{ S_i :  c_i\not= 0\} |$.
Hence the maximum harmonic mean set cover problem is the standard minimum set cover problem.

This problem is among the most studied $\NP$-hard problems. It has long been known to be approximable within a factor $H_{\max_i |S_i|}$ with the greedy algorithm. The first proof is from Johnson~\cite{J74}. Lov\'asz~\cite{L75} obtained the same factor with a different method. Later, Chv\'atal extended the result to the weighted set cover problem~\cite{C79}, in which the subsets $S_i$ have nonuniform costs. A number of papers show that the logarithmic approximation guarantee is likely to be optimal. Lund and Yannakakis~\cite{LY94} first proved that the problem is not approximable within $\log n / 4$ unless $\NP\subseteq\DTIME(n^{\text{polylog}(n)})$. This result has been improved to $(1-o(1))\ln n$ by Feige~\cite{F98}, under the hypothesis $\NP\not\subseteq\DTIME(n^{O(\log \log n)})$. Raz and Safra~\cite{RS97}, and Alon, Moshkovitz, and Safra~\cite{AMS06} proved inapproximability results for factors of the form $c\ln n$ for some constant $c$ under the hypothesis $\P\not=\NP$. These results are consequences of new PCP characterizations of $\NP$.

\paragraph{Minimum Entropy Set Cover.}
Let us now consider the geometric mean version: $\max_{\phi} \left( \prod_{v\in V}a_v\right)^{\frac1n}$.
We relate this mean to the {\em entropy} of the discrete probability distribution found by dividing each part size by $n$:
\begin{eqnarray*}
- \sum_{i=1}^k  \frac{c_i}n \log \frac{c_i}n & = & - \sum_{v\in V} \frac1n \log\frac{a_v}n \\
& = & \log n - \frac1n \sum_{v\in V} \log a_v \\
& = & \log n-\log M_0 (\{ a_v : v\in V\} ). \label{eqn:entgeom}
\end{eqnarray*}
Thus the maximum geometric mean set cover problem is equivalent to the problem of minimizing the entropy of the partition.
This problem is known as the {\em minimum entropy set cover} problem. It has been introduced by Halperin and Karp~\cite{HK05},
and has applications in the field of computational biology. They proved that the problem was approximable within a constant additive term with the greedy algorithm. Improving on this work, Cardinal, Fiorini, and Joret~\cite{CFJ06} provided a simple analysis showing that the constant was at most $\log_2 e\simeq 1.4427$ bits, and that this was the smallest additive error achievable in polynomial time, unless $\P =\NP$. The minimum entropy vertex cover~\cite{CFJ08} and minimum entropy graph coloring~\cite{CFJ05} problems, which are special cases of minimum entropy set cover, have been studied by the same authors.

\paragraph{Maximum-Edge Clique Partition.}
In a recent publication~\cite{DJLLP07}, Dessmark, Jansson, Lingas, Lundell, and Persson studied the {\em maximum-edge clique partition} (Max-ECP) problem. In this problem, we aim to partition a graph $G$ into cliques in order to maximize the number of edges whose endpoints are in the same clique of the partition. This is an implicit set cover problem, in which the subsets $S_i$ are the cliques of $G$, and the function to maximize is: 
\begin{equation*}
\sum_{i=1}^k {c_i\choose 2} = \frac 12 \left(-n+\sum_{i=1}^k c_i^2\right) = \frac n2 \left( M_1(\{ a_v : v\in V\} ) - 1 \right).
\end{equation*}
Thus the problem can be seen as an implicit maximum $p$-mean set cover problem for $p=1$. They show that the problem is $2$-approximable on perfect graphs using the greedy algorithm, and that it is not approximable within a factor $n^{1-O(1/(\log n)^{\gamma})}$ for some constant $\gamma$ in polynomial time unless $\NP\subseteq \ZPTIME (2^{(\log n)^{O(1)}})$. 

\paragraph{Max-Max and Max-Min Set Cover.}
When $p\to\infty$, the maximum $p$-mean set cover problem involves finding a cover in which the largest part has maximum size. This is a trivial problem, unless the subsets in $\bigs$ are not given explicitly, like in the graph coloring problem. For $p\to -\infty$, the problem is that of maximizing the size of the smallest part, thus solving $\max_{\phi} \min_{v\in V} a_v = \max_{\phi} \min_{i:c_i\not= 0} c_i$. This problem seems much more challenging. We will refer to it as the max-min set cover problem.

\subsection*{Our results}

We show in section~\ref{sec:approx} that for any $p\in{\mathbb R}^+$, the maximum $p$-mean set cover problem is approximable within a factor of $(p+1)^{1/p }$. This factor is less than $e$ for all positive values of $p$, hence this can be seen as a {\em robust} $e$-approximation for all $p$-means with positive $p$. This result generalizes the approximability results of Cardinal et al.~\cite{CFJ06} for the case $p\to 0$, and of Dessmark et al.~\cite{DJLLP07} for $p=1$. We also prove that this is the best we can achieve in polynomial time unless $\P = \NP$, using a powerful reduction due to Feige et al.~\cite{F98,FLT04}. When $p$ is negative, we show that the performance of the greedy algorithm degrades. We give an inapproximability result for max-min set cover.

Graph coloring problems can be seen as implicit set cover problems in which the subsets $S_i$ are the maximal independent sets of the graph. 
The subsets are not given explicitly, which would cause an exponential blowup in the problem size, but rather implicitly, from the graph structure.
We define the maximum $p$-mean graph coloring problem in this natural way. Special cases of the maximum $p$-mean graph coloring problem include the standard minimum coloring problem ($p=-1$), the minimum entropy coloring problem~\cite{CFJ05} ($p\to 0$), the maximum-edge clique partition problem~\cite{DJLLP07} ($p=1$), and the maximum independent set problem ($p\to +\infty$). In Section~\ref{sec:coloring} we give approximability and inapproximability results for this problem. 

The maximum $p$-mean set cover problem involves maximizing the sum of the $(p+1)$th powers of the part sizes, as can be seen in equation~(\ref{eq:genmean}). In section~\ref{sec:gen}, we consider weighted instances, and a further generalization of the set cover problem, in which we aim at minimizing the sum of some concave function of the part sizes. We give a closed form of the approximation ratio achieved by the greedy algorithm for this general class of problems, and apply this result to the case of the Rent-or-Buy set cover problem~\cite{FHN07b}.

\subsection*{Related works}

\paragraph{Minimum sum set cover.}
In the minimum sum set cover problem we aim to find an ordering of the subsets that minimizes the average
{\em cover time} of an element of the ground set, where the cover time of an element is the index of the first subset covering it.
This problem was first considered in its graph coloring version~\cite{BBHST98}.
Feige, Lov\'asz, and Tetali~\cite{FLT04} gave an elegant proof of the fact that greedy is a 4-approximation algorithm, and 
that this was the best one could hope for unless $\P =\NP$. They also studied the related minimum sum vertex cover problem, for which
they provided a 2-approximation algorithm.

\paragraph{Generalizations of minimum sum set cover.}
Munagala, Babu, Motwani, and Widom~\cite{MBMW05} introduced the {\em pipelined set cover} problem.
In this problem, we aim to find an ordering of the subsets in $\bigs$ that minimizes the $L_p$-norm of the vector
$(R_i)$, where $R_i$ is the number of elements that are not contained in
any of the first $(i-1)$ subsets. For $p=1$, the problem is equivalent to the minimum sum set cover problem.
They generalize the technique of Feige et al.~\cite{FLT04} to prove a $4^{\frac 1p}$-approximation.

More recently, Golovin, Gupta, Kumar, and Tangwongsan~\cite{GGKT07} considered another minimum
$L_p$-norm set cover problem. This variant involves finding an ordering of the subsets minimizing the $L_p$-norm
of the cover time vector. This problem is a simultaneous generalization of the minimum set cover problem and the 
minimum sum set cover problem. They prove that the greedy algorithm provides a $O(p)$-approximate solution, 
and that this is the best possible, up to a constant factor, unless $\NP\subseteq \DTIME (n^{O(\log\log n)})$.

\paragraph{Graph Coloring.}
The greedy algorithm for set cover translates to the MaxIS algorithm for graph coloring, in which a maximum independent set is iteratively chosen as new color class. This algorithm has in particular been analyzed for the minimum sum~\cite{BBHST98} and minimum entropy~\cite{CFJ05,CFJ06} graph coloring problems.

Recently, Fukunaga, Halld\'orsson, and Nagamochi~\cite{FHN08} initiated the study of a very general family of minimum cost graph coloring problems,
similar to what we propose in section~\ref{sec:gen}. They proved that any minimum cost graph coloring problem in this family is 4-approximable on weighted interval graphs, provided that the cost function is both monotone and concave. The proposed algorithm iteratively removes a maximum $i$-colorable subgraph, where $i$ is doubled at each iteration. 

In another recent contribution, Fukunaga, Halld\'orsson, and Nagamochi~\cite{FHN07b} introduced the {\em Rent-or-Buy} coloring problem in vertex-weighted graphs, in which the cost of a color class is the minimum between 1 and the total weight of the class. This models situations in which each color class has to be paid for either by ``buying" it for a fixed cost, or ``renting" it for a price proportional to its size. They gave, among other results, a 2-approximation for this problem in perfect graphs. We consider the set cover version of this problem in section~\ref{sec:gen}.

\paragraph{Clique Partitioning with Value-polymatroidal Costs.}
Gijswijt, Jost, and Queyranne~\cite{GJQ07} recently studied clique partitioning problems with {\em value-polymatroidal} cost functions.
A function $f$ over the subsets of $V$ is said to be value-polymatroidal whenever $f (\emptyset ) = 0$, $f$ is non-decreasing and for every subsets 
$S$ and $T$ with $f(S)\geq f(T)$, and every $u$ in $V \setminus (T\cup S)$, the inequality 
$f(S+u) - f(S) \leq f(T+u) - f(T)$ holds. They define the cost of a clique partition as the sum of the cost of each clique.
They prove, among other results, that this problem is solvable in polynomial time on interval graphs.

\paragraph{Minimum $L_p$-norm problems.}
Azar, Epstein, Richter, and Woeginger~\cite{AzarERW02} studied approximation algorithms for a scheduling problem in which we aim to {\em minimize} the $L_p$-norm of the part sizes. A similar problem has been studied by Azar and Taub~\cite{AzarT04}, who proposed all-norm approximation algorithms. Although similar in spirit, the goal is different than ours, since we instead seek the most ``nonuniform" distribution, with {\em maximum} $L_p$-norm.\\

A number of other problems with general cost functions have been studied, such as facility location~\cite{HMM03}. Due to space constraints, we do not give more details here.

\section{Approximability}
\label{sec:approx}

\begin{lemma}
\label{lem:approx}
The maximum $p$-mean set cover problem for $p\in{\mathbb R}$ is approximable in polynomial time within a factor of 
\begin{equation}
\label{eq:ratio}
\left(\frac{n^{p+1}}{\sum_{j=1}^n j^p}\right)^{\frac 1p}.
\end{equation}
\end{lemma}
\begin{proof}
We consider an optimal cover $\phi_{\opt}$, and a part $C_i=\phi_{\opt}^{-1}(S_i)$ in this cover, of size $|C_i|=c_i$. We define $a'_v := |\phi ^{-1} (\phi  (v)) |$ for the cover $\phi $ returned by the greedy algorithm. 

We first suppose that $p\geq 0$, and give a lower bound on the value of the cover $\phi$ restricted to $C_i$. We do so by examining the elements of $C_i$ in the order in which they are covered by the greedy algorithm, breaking ties arbitrarily. The first covered element $v_1\in C_i$ must belong to a part of size at least $c_i$ in $\phi $, since $C_i$ can be chosen as a new part, and the greedy algorithm chooses the largest part. Hence $a'_{v_1}\geq c_i$. Similarly, the second element $v_2$ of $C_i$ that is covered by greedy must belong to a class of size at least $c_i-1$. Hence $a'_{v_2}\geq c_i-1$. In general, for the $k$th element $v_k$ covered by the greedy algorithm, $a'_{v_k}\geq c_i - k + 1$. Thus we have 
\begin{equation}
\label{eqn:lbci}
\sum_{v\in C_i} (a'_v)^p \geq \sum_{j=1}^{c_i} j^p. 
\end{equation}
Letting $a_v:= |\phi_{\opt} ^{-1} (\phi_{\opt} (v)) |$, the corresponding value for $\phi_{\opt}$ is $\sum_{v\in C_i} a_v^p = c_i^{p+1}$, hence we get
the following upper bound
\begin{equation}
\label{eqn:upratio}
\frac{\sum_{v\in C_i} a_v^p}{\sum_{v\in C_i} (a'_v)^p} \leq 
\frac{{c_i}^{p+1}}{\sum_{j=1}^{c_i} j^p}.
\end{equation}
This ratio is increasing with $c_i$, and holds for all the parts $C_i$ of $\phi_{\opt}$. Letting $c_i=n$ and taking the $p$th root gives the result.

A similar reasoning holds for $p<0$, with the direction of inequalities~(\ref{eqn:lbci}) and (\ref{eqn:upratio}) reversed.
\qed\end{proof}

\begin{figure}
\begin{center}
\includegraphics[scale=.4]{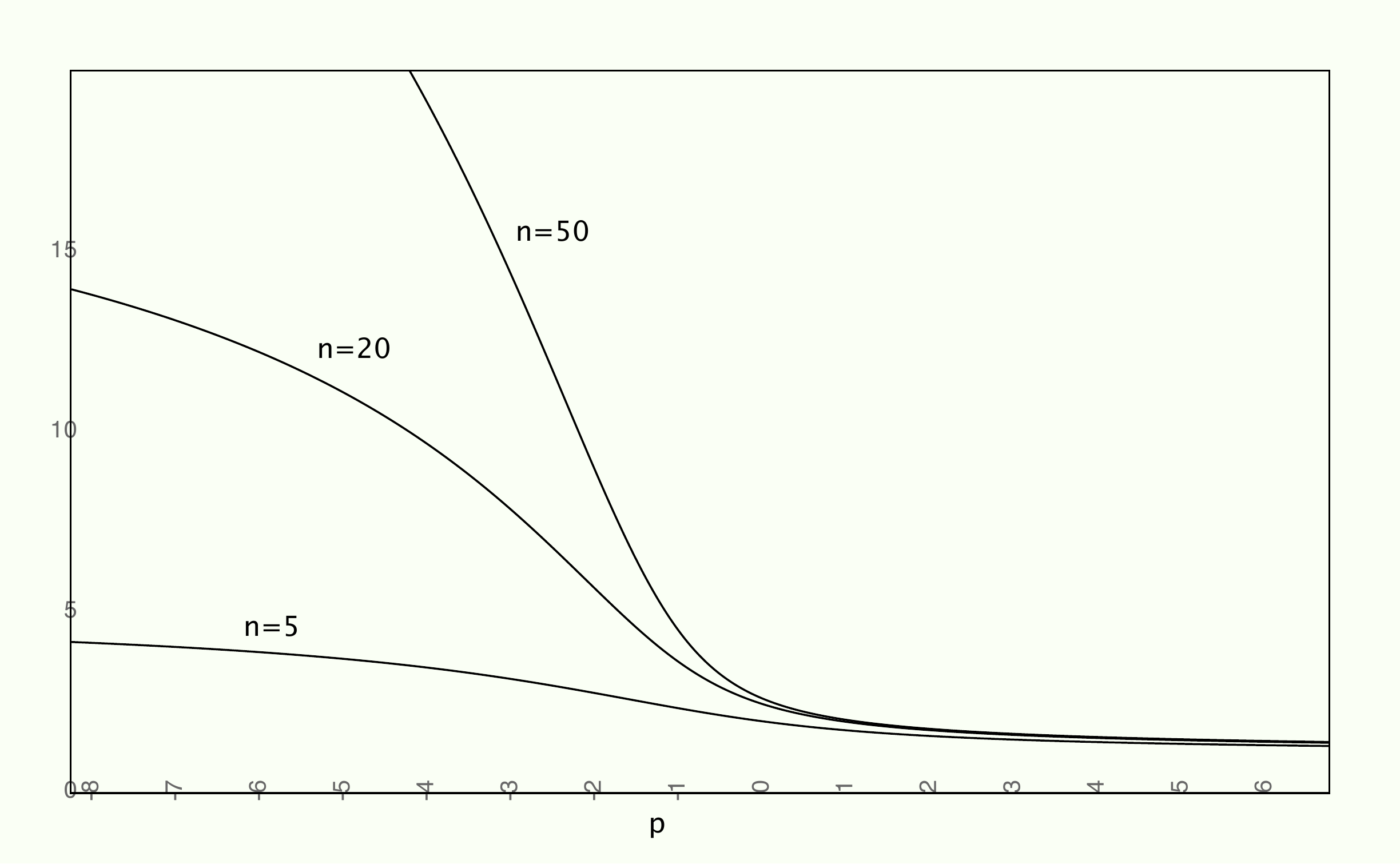}
\end{center}
\caption{\label{fig:ratios}Approximation ratios for the greedy algorithm.}
\end{figure}

The approximation ratios for various values of $p$ and $n$ are given in Fig.~\ref{fig:ratios}. We next give a constant upper bound on the approximation ratio in the case $p\geq 0$. We need the following lemma.
\begin{lemma}
\label{lem:simplesum}
For $p\in{\mathbb R}^+$ and $n\in \mathbb N$, $$\sum_{j=1}^n j^p \geq \frac{n^{p+1}}{p+1}.$$
\end{lemma}
\begin{proof}
The inequality holds for $p=0$. For $p>0$, it can be checked graphically that approximating the sum by an integral yields a lower bound:
$$
\sum_{j=1}^n j^p = \sum_{j=0}^n j^p > \int_0^n x^p dx = \frac{n^{p+1}}{p+1}.
$$
\qed\end{proof}

Combining lemmas~\ref{lem:approx} and \ref{lem:simplesum} proves the following theorem. Tightness can be proved using known tight examples for special cases (see for instance~\cite{CFJ06}).
\begin{theorem}
\label{thm:approx}
The maximum $p$-mean set cover problem is approximable in polynomial time within a factor of $(p+1)^{\frac 1p }$ for $p\in{\mathbb R}^+$. This bound is asymptotically tight.
\end{theorem}

Note that $\lim_{p\to +\infty} (p+1)^{\frac 1p } = 1$, hence in the case of $p\to +\infty$, the approximation ratio is equal to $1$. This formalizes the trivial observation that if our goal is to maximize the size of the largest part, then the greedy algorithm returns an optimal solution.
Also, $\lim_{p\to 0} (p+1)^{\frac 1p } = e$, which proves that the greedy algorithm approximates the minimum entropy set cover within an additive term
of $\log e$ bits. This was shown by Cardinal, Fiorini, and Joret~\cite{CFJ06}. Finally, for $p=1$, the greedy algorithm returns a 2-approximation. A proof of this result was given by Dessmark, Jansson, Lingas, Lundell, and Persson~\cite{DJLLP07}.\\

We now turn to the case $p<0$. We know that the greedy algorithm approximates the problem for $p=-1$ within a logarithmic factor. The following result shows that the performance of greedy degrades dramatically as $p$ becomes smaller. 

\begin{theorem}
\label{thm:approxq}
The maximum $p$-mean set cover problem is approximable in polynomial time within a factor of $n^{1-\frac 1q} \zeta (q)^{\frac 1q}$ for any real $p = -q <-1$, where $\zeta (q)=\sum_{j=1}^{\infty} j^{-q}$ is the Riemann zeta function. 
\end{theorem}
\begin{proof}
We consider expression~(\ref{eq:ratio}) in lemma~\ref{lem:approx} and replace $p$ by $-q$:
\begin{equation}
\left(\frac{n^{1-q}}{\sum_{j=1}^{n} j^{-q}} \right)^{-\frac 1q}   = 
					\left(\frac{\sum_{j=1}^{n} j^{-q}}{n^{1-q}} \right)^{\frac 1q} 
					 \leq \left( \frac{\zeta (q) }{n^{1-q}} \right)^{\frac 1q} 
					 = n^{1-\frac 1q} \zeta (q)^{\frac 1q}.
\end{equation}
\qed\end{proof}

The bound is asymptotically tight if we replace $n$ by $\max_i |S_i|$. Note that we need $q>1$, otherwise the Dirichlet series defining the zeta function does not converge. In particular, when $q=1$ (and thus $p=-1$), we have the harmonic series, which is the approximation ratio for the minimum set cover problem.

An interesting special case is when $p=-2$. This means that the cost of a part of size $c_i$ in the cover is $1/c_i$. In that case, the approximation ratio of the greedy algorithm becomes 
\begin{equation}
n^{1-\frac 12} \zeta (2)^{\frac 12} = \pi \sqrt{\frac{n}{6}} .
\end{equation}

We now show that the approximability result in theorem~\ref{thm:approx} for positive values of $p$ is the best we can hope for, unless $\P = \NP$. 
We need the following lemma, which is a simple consequence of the convexity of the function $f(x)=x^{p+1}$.\\

Consider two sorted sequences $c_1\geq c_2 \geq\ldots\geq c_k$ and $c'_1\geq c'_2\geq\ldots\geq c'_k$. We say that
$(c_i)$ {\em dominates} $(c'_i)$ if 
\begin{equation}
\sum_{i=1}^j c_i \geq \sum_{i=1}^j c'_i\ \forall\ j\in\{ 1,2,\ldots ,k\}.
\end{equation}
\begin{lemma}
\label{lem:jungle}
If $(c_i)$ {\em dominates} $(c'_i)$, then for any $p\in\mathbb R^+$,
\begin{equation}
\sum_{i=1}^k c^{p+1}_i \geq \sum_{i=1}^k \left( c'_i\right)^{p+1} .
\end{equation}
\end{lemma}

\begin{theorem}
\label{thm:inapprox}
It is $\NP$-hard to approximate the maximum $p$-mean set cover problem within a factor less than $(p+1)^{\frac 1p }$ for $p\in{\mathbb R}^+$.
\end{theorem}
\begin{proof}
Feige, Lovasz, and Tetali~\cite{FLT04} gave a procedure for transforming a 3SAT-6 formula into a set system $(V,\bigs )$ with the following properties:
\begin{itemize}
\item each subset $S_i\in\bigs$ has size $n/t$ for a certain parameter $t$,
\item if the formula is satisfiable, then there exists an exact cover of $V$ with $t$ subsets,
\item if the formula is $\delta$-satisfiable, that is, if at most a fraction $\delta$ of the clauses can be satisfied, then every $i$ subsets of $\bigs$ 
cover at most a fraction $(1-(1-1/t)^i)-\epsilon$ of the elements of $V$, for $i\in\{1,2,\ldots ,at\}$ and any choice of the constants $\epsilon > 0$ and
$a>0$.
\end{itemize}
Given a formula known to be either satisfiable or $\delta$-satisfiable, the problem of distinguishing between the two is $\NP$-hard~\cite{FLT04}. Using the
transformation above, we show that a polynomial algorithm with an approximation ratio less than $(p+1)^{\frac 1p }$ for maximum $p$-mean set cover would solve this problem.

If the formula is satisfiable, then $V$ can be covered by exactly $t$ disjoint sets of $\bigs$. From Lemma~\ref{lem:jungle}, this is the optimal solution. The part sizes $c_i$ in this solution satisfy
\begin{equation}
\label{eq:exact}
\sum_{i=1}^k \left( \frac{c_i}{n}\right)^{p+1}=\sum_{i=1}^t \left( \frac{1}{t}\right)^{p+1} = \frac{1}{t^p} .
\end{equation}

We now suppose the formula is only $\delta$-satisfiable. We consider the distribution in which the $i$th part covers a fraction 
$$
\left( 1-\left( 1-1/t\right) ^i\right) - \left( 1-\left( 1-1/t\right)^{i-1}\right) = \frac{1}{t}\left(1-\frac{1}{t}\right)^{i-1}
$$ 
of the elements of $V$, for $i\in\{1,2,\ldots ,at\}$, and the remaining parts cover exactly a fraction $\frac{1}{t}\left(1-\frac{1}{t}\right)^{at}$ each. We denote
by $r$ the number of remaining parts, so that the sum of the fractions equals $1$. From Lemma~\ref{lem:jungle} and the properties of the reduction, this distribution dominates all other achievable distributions. Therefore the following upper bound holds.
\begin{eqnarray}
\sum_{i=1}^k \left( \frac{c_i}{n}\right)^{p+1} & \leq & \sum_{i=0}^{at}{\left(\frac{1}{t}\left(1-\frac{1}{t}\right)^i \right)^{p+1}} + r \left(\frac{1}{t}\left(1-\frac{1}{t}\right)^{at} \right)^{p+1} \\
& \simeq & \frac{1}{t^{p+1}} \sum_{i=0}^{at} e^{-\frac{(p+1)i}{t}} + \frac{r}{t^{p+1}} e^{-a(p+1)}.  \label{eqn:tmp}
\end{eqnarray}
We can approximate the sum by an integral :
\begin{equation}
\sum_{i=0}^{at} e^{-\frac{(p+1)i}{t}}  \simeq  \int_0^{at} e^{-\frac{(p+1)x}{t}}\cdot dx =  \frac t{p+1} \left( 1 - e^{-a(p+1)} \right) .
\end{equation}
The value $r$ is the number of parts of size $\frac{1}{t}\left(1-\frac{1}{t}\right)^{at}\simeq \frac 1t e^{-a}$ needed to cover a fraction 
$1-\sum_{i=0}^{at} \frac{1}{t}\left(1-\frac{1}{t}\right)^i \simeq e^{-a}$ of the elements. Thus $r\sim t$, and 
$$
\frac{r}{t^{p+1}} e^{-a(p+1)} \simeq \frac 1{t^p} e^{-a(p+1)}.
$$
Note that since the constant $t$ can be assumed to be arbitrary large~\cite{FLT04}, the approximations above are arbitrarily accurate. 
Hence expression~(\ref{eqn:tmp}) can be made arbitrarily close to:
\begin{equation}
\label{eq:delta}
\frac{1}{t^p} \left( \frac{1}{p+1} \cdot \left(1 - e^{-a(p+1)} \right) + e^{-a(p+1)}\right) .
\end{equation}
Now by choosing $a$ sufficiently large, the ratio between (\ref{eq:delta}) and (\ref{eq:exact}) can be made arbitrary close to $p+1$. The gap between the 
$p$-means is obtained by taking the $p$th root.
\qed\end{proof}

In the case $p\to 0$, the above inapproximability proof shows that the additive $\log e$ error term is best possible (unless $\P =\NP$) for the minimum entropy set cover problem. This was also shown previously by Cardinal, Fiorini, and Joret~\cite{CFJ06}.\\ 

Although we do not have a precise inapproximability threshold for negative values of $p$, we can prove the following result for $p\to -\infty$. That is the max-min set cover problem, in which we aim to maximize the size of the smallest part.

\begin{theorem}
\label{thm:inapproxmaxmin}
It is $\NP$-hard to approximate the max-min set cover problem within any constant factor.
\end{theorem}
\begin{proof}
The proof uses the same reduction as the proof of theorem~\ref{thm:inapprox}. We consider set systems $(V,\bigs )$ constructed from a 3SAT-6 formula,
such that there exists an exact cover with $t$ parts of size $\frac nt$ if the formula is satisfiable, and every $i$ subsets of $\bigs$ 
cover at most a fraction $(1-(1-1/t)^i)-\epsilon$ of the elements, for $i\in\{1,2,\ldots ,at\}$, if the formula is $\delta$-satisfiable. But this means that in the latter case, at least $at$ subsets are needed to cover $V$. This implies that there is a part of size at most $\frac n{at}$. Since $a$ can be chosen arbitrarily greater than any constant, the gap can be made arbitrarily large.
\qed\end{proof}

\section{Graph Coloring}
\label{sec:coloring}

We now define the graph coloring variant of the maximum $p$-mean set cover problem.

\begin{definition}[Maximum $p$-mean graph coloring]
Given a simple, undirected graph $G=(V,E)$, find an assignment $\phi: V \mapsto \mathbb{N}$
of colors to vertices such that adjacent vertices receive different colors, and
$ M_p (\{ a_v : v\in V\} )$ is maximized, where $a_v :=|\phi^{-1}(\phi (v))|$ and $M_p$ is the generalized mean with parameter $p$.
\end{definition}

The greedy algorithms extends naturally to what is referred to as the MaxIS algorithm, in which a maximum independent set is iteratively 
removed from the graph. This procedure can run in polynomial time only if at each step we can find a maximum independent set in polynomial
time. This is true for large families of graphs, such as perfect graphs~\cite{GLS93}, and claw-free graphs~\cite{ND01}. 
We thus have the following corollary of theorem~\ref{thm:approx} (the proof of tightness is omitted).

\begin{corollary}
The maximum $p$-mean graph coloring problem restricted to perfect or claw-free graphs is approximable in polynomial time within a factor of $(p+1)^{\frac 1p }$ for $p\in{\mathbb R}^+$. This bound is asymptotically tight.
\end{corollary}

It may happen that we only have an approximate algorithm for the maximum independent set problem. Then the following result applies.
Proofs are given in appendix~\ref{app:rhoappx}.

\begin{theorem}
\label{thm:rhoappx}
If the maximum independent set problem can be approximated within a factor $\rho$ in polynomial time,
then the maximum $p$-mean graph coloring problem is approximable within a factor of $\rho (p+1)^{\frac 1p }$ in polynomial time.
\end{theorem}
\begin{corollary}
\label{cor:mecappx}
The minimum entropy coloring problem~\cite{CFJ05} is approximable in polynomial time within an additive error of $\log_2 (\Delta + 2) - 0.14226$ on graphs with maximum degree $\Delta$.
\end{corollary}

In the max-min graph coloring problem, that is when $p\to -\infty$, we aim to maximize the size of the smallest color class. Using a recent polynomial algorithm from 
Kierstead and Kostochka to construct equitable $\Delta + 1$-colorings~\cite{KK08}, we have the following approximability result.
\begin{corollary}
The max-min graph coloring problem can be approximated in polynomial time within a factor $\left( 1+O\left(\frac 1n\right)\right)\frac{\Delta + 1}{\chi}$ on graphs of order $n$, maximum degree $\Delta$, and chromatic number $\chi$.
\end{corollary}

The maximum independent set problem is the special case of minimum $p$-mean coloring in which $p\to +\infty$. It is therefore 
not surprising that the general coloring problem is not well approximable for any positive value of $p$, as the following lemma shows.
\begin{lemma}
\label{lem:isred}
If the maximum independent problem set cannot be approximated in polynomial time within $n^{1-\epsilon}$ for some $\epsilon =\epsilon (n)$, then the 
maximum $p$-mean graph coloring problem with $p\in{\mathbb R}^+$ cannot be approximated in polynomial time within $n^{1-\left( 2+\frac 1p\right)\epsilon}$.
\end{lemma}
\begin{proof}
If the maximum independent set cannot be approximated within $n^{1-\epsilon }$, then we can safely assume that this
holds for graphs having an independent set of size $\alpha \geq n^{1-\epsilon }$. In such a graph, we consider the 
coloring obtained with a $n^{1-t\epsilon }$-approximation algorithm for maximum $p$-mean coloring, for some constant $t$ to be fixed later. 

The optimal solution in this graph has value at least $\left( \alpha^{p+1}\right)^{\frac 1p}$. Thus the
value $A$ of the coloring satisfies
\begin{equation}
A\geq \frac{\left( \alpha^{p+1} \right)^{\frac 1p}}{n^{1-t\epsilon }}.
\end{equation}
We now consider the largest color class in this coloring, and denote its size by $h$. We then get the following upper bound on $A$:
\begin{equation}
A\leq \left( \frac nh h^{p+1}\right)^{\frac 1p} = h n^{\frac 1p}.
\end{equation}

Putting this together, we obtain
\begin{eqnarray}
h n^{\frac 1p} & \geq & \frac{\left( \alpha^{p+1} \right)^{\frac 1p}}{n^{1-t\epsilon }} 
\geq \frac{\left( n^{(1-\epsilon )(p+1)} \right)^{\frac 1p}}{n^{1-t\epsilon }} \\
h & \geq & n^{\left( t-1-\frac 1p\right) \epsilon }.
\end{eqnarray}
%
Letting $t=2+\frac 1p$, we obtain an independent set of size at least $n^\epsilon $, which is a $n^{1-\epsilon }$-approximation for the maximum
independent set problem, a contradiction.
\qed\end{proof}

Applying this lemma and using a result from Khot~\cite{K01}, we obtain the following.
\begin{theorem}
The maximum $p$-mean graph coloring problem, for $p\in{\mathbb R}^+$, is not approximable in polynomial time within a factor 
$n^{1-O(1/(\log n)^{\gamma})}$ for some constant $\gamma$ unless $\NP\subseteq \ZPTIME (2^{(\log n)^{O(1)}})$.
\end{theorem}

A similar result for $p\to 0$ was proved by Cardinal et al.~\cite{CFJ05}. The special case $p=1$ was proved by Dessmark et al.~\cite{DJLLP07}.\\

We end our discussion of the graph coloring problems with the equivalent problem in the complement of the graph $G$, which we call the {\em maximum $p$-mean clique partition} problem. The Max-ECP problem corresponds to the special case $p=1$. Gijswijt, Jost, and Queyranne~\cite{GJQ07} provided a $O(n^3)$ dynamic programming algorithm for finding a partition of interval graphs in cliques that minimizes the sum of a value-polymatroidal cost. Unfortunately, our objective function do not fall in that class, since the equivalent minimization problem involves minimizing a concave decreasing cost function, and value-polymatroidal functions must be non-decreasing. However, the correctness of their dynamic programming solely relies on the fact that an optimal partition always contain a maximal clique. This is true in our case as well, at least for $p>0$, and is a consequence of lemma~\ref{lem:jungle}. Thus the algorithm can be applied and we get the following results.
\begin{theorem}
The maximum $p$-mean clique partition problem with $p\in{\mathbb R}^+$ can be solved in $O(n^3)$ time on interval graphs.
\end{theorem}
\begin{corollary}
The Max-ECP problem~\cite{DJLLP07} can be solved in $O(n^3)$ time on interval graphs.
\end{corollary}

\section{Further Generalizations}
\label{sec:gen}

\paragraph{Weighted variant.}
We first observe that theorems~\ref{thm:approx} and \ref{thm:inapprox} also hold for a weighted version of the minimum $p$-mean set cover problem. In this problem, the elements of $v$ have a weight $w(v)$. The objective function is the same, except that $a_v$ is now defined as $w(\phi^{-1}(\phi (v)))$. We can observe that the approximability proofs above still hold using a simple reduction for integer weights. Given a weighted instance, we can transform it into an unweighted instance by replacing each element $v\in V$ by $w(v)$ copies of it, each belonging to the same subsets as $v$. Then each copy of the duplicated elements must belong to the same part of the (greedy or optimal) solution. Otherwise, from lemma~\ref{lem:jungle}, some elements can be reassigned so that the $p$-mean increases. The argument extends to rational and, by continuity, real weights.

\paragraph{General costs.}
Following the definition of Fukunaga, Halld\'orsson, and Nagamochi~\cite{FHN08} for minimum cost colorings, we now consider a much more general family of set cover problems. In these problems, we aim to minimize a sum of some concave function $f(c_i)$ of the part sizes. The functions $f$ are concave in the sense that they are discrete restrictions of concave functions $f:{\mathbb R^+}\mapsto{\mathbb R}$. We also assume $f(0)=0$. Setting $f(c_i)=-c_i^{p+1}$, for instance, yields a problem similar to the maximum $p$-mean set cover problem, without the $1/p$ exponent. The definition of this new family is as follows. 

\begin{definition}[Set cover with general costs]
Given an $n$-element ground set $V$ and a collection $\bigs  =\{S_1,\dots,S_k\}$ 
of subsets of $V$ whose union is $V$, find a cover $\phi: V \mapsto \bigs $ that minimizes $\sum_{i=1}^k f(c_i)$,
where $c_i:=|\phi^{-1}(S_i)|$ and $f$ is a concave function.
\end{definition}

Concavity implies that we seek a distribution of the part sizes that is as unbalanced as possible. In particular, the following generalization of 
lemma~\ref{lem:jungle} holds. 
\begin{lemma}
\label{lem:genjungle}
Given two nonincreasing sequences $(c_i)$ and $(c'_i)$, such that $(c_i)$ dominates $(c'_i)$, and a concave function $f$, 
we have $\sum_{i=1}^k f(c_i) \leq \sum_{i=1}^k f(c'_i)$.
\end{lemma}

Although the approximation ratio obtained with the greedy algorithm depends on the function $f$, we can give
a simple expression of it.
\begin{theorem}
\label{thm:genapprox}
The set cover problem with general costs can be approximated in polynomial time within a factor of 
$$\max \left\lbrace \frac{1}{f(c)}\sum_{j=1}^c \frac{f(j)}{j} : 1\leq c\leq \max_i |S_i|\right\rbrace .$$
\end{theorem}
\begin{proof}(sketch)
Given a solution $\phi$, we associate to each element $v\in V$ the cost $\frac{f(a_v)}{a_v}$, where $a_v=|\phi^{-1}(\phi (v))|$ as before.
The cost of this solution is the sum $\sum_{v\in V} \frac{f(a_v)}{a_v}$. Using concavity, we can bound this sum in a greedy solution as in the proof of lemma~\ref{lem:approx}: we show that the sum over the elements in a part of size $c$ in the optimal solution is at most $\sum_{j=1}^c f(j)/j$. The ratio follows.
\qed\end{proof}
Note that we retrieve the approximation ratio $H_n$ of minimum set cover by setting $f(c)=1$ if $c>0$ and $f(0)=0$. This result also encompasses our analyses of the approximability of minimum entropy and maximum $p$-mean set cover.\\

We now give an application of this result to a new problem. In this problem, we suppose that the cost of assigning an element of $V$ to a subset $S_i$ is 1 if $S_i$ covers a lot of elements, but is proportional to its size if the fraction of elements covered by $S_i$ is small. More precisely, if the fraction $c_i/n$ of elements covered by $S_i$ is greater than some constant $\beta <1$, then the incurred cost is $\frac {c_i}n / \beta$. Otherwise, the cost is 1. Thus $\beta$ defines a breakpoint, above which it is less costly to ``buy" the subset than ``rent" it. Hence we define the {\em Rent-or-Buy} set cover problem as the set cover problem with the following cost function:
$$f(c) = 
\begin{cases} 
c/(\beta n) & \text{if } c \leq \beta n, \\
1 & \text{otherwise.} 
\end{cases}$$
This models situations in which for instance jobs are assigned to machines, and machines can be either bought or rented. The model was introduced recently by Fukunaga, Halld\'orsson, and Nagamochi as a graph coloring problem~\cite{FHN07b}. The original description of the Rent-or-Buy model was on a weighted graph, and the coloring problem was to find a coloring minimizing the sum of the values $\min \{1, w(C_i)\}$ over all color classes $C_i$, where $w(C_i)$ is the sum of the weights of the vertices in $C_i$. From our reduction of weighted instances described above, this is equivalent to our problem with
$\beta = \frac 1{w(V)}$.

\begin{corollary}
The Rent-or-Buy set cover problem is approximable in polynomial time within a factor of $1-\ln\beta$. 
\end{corollary}
\begin{proof}
We let $t=\beta n$. Let us first suppose that $c\leq t$. Then we have
\begin{equation}
\frac{1}{f(c)} \sum_{j=1}^c \frac{f(j)}{j} = \frac{t}{c} \sum_{j=1}^c \frac{1}{t}  = 1.
\end{equation}
Otherwise, if $c>t$, we have
\begin{equation}
\frac{1}{f(c)} \sum_{j=1}^c \frac{f(j)}{j}  =  1 + \sum_{j=t + 1}^c \frac{1}{j} 
 =  1 + H_c - H_t  \leq  1 + H_n - H_t  \leq  1 - \ln\beta .
\end{equation}
Hence from Theorem~\ref{thm:genapprox}, this is the worst-case approximation ratio achieved by the greedy algorithm.
\qed\end{proof}

Since the greedy algorithm can be implemented to run in polynomial time on perfect or claw-free graphs, we obtain the following result on the Rent-or-Buy graph coloring problem.
\begin{corollary}
The Rent-or-Buy coloring problem is approximable in polynomial time within a factor of $1+\ln w(V)$ on perfect or claw-free graphs.
\end{corollary}
This improves on the 2-approximation algorithm~\cite{FHN07b} when the overall weight $w(V)$ does not exceed $e$.

\bibliography{cost_coverings}
\bibliographystyle{plain}

\appendix

\section{Proof of theorem~\ref{thm:rhoappx} and corollary~\ref{cor:mecappx}}
\label{app:rhoappx}

The proof is similar to that of lemma~\ref{lem:approx}. We consider the approximate MaxIS algorithm in which a $\rho$-approximate maximum independent set is chosen at each step. We consider a class $C_i$ in an optimal coloring, of size $c_i$. The first vertex $v_1$ of $C_i$ that is colored by the approximate MaxIS algorithm will be assigned a value $a'_{v_1}$ at least $c_i/\rho$, since there exists an independent set of size $c_i$ in the current graph. By iterating this argument, we obtain that $\sum_{v\in C_i} (a'_v)^p \geq \frac{1}{\rho^p}\sum_{j=1}^{c_i} j^p$. In the optimal coloring, the value of this color class is $c_i^{p+1}$. Hence the ratio is at most
\begin{equation}
\left(\frac{n^{p+1}}{\frac{1}{\rho^p}\sum_{j=1}^n j^p}\right)^{\frac 1p} = \rho \left(\frac{n^{p+1}}{\sum_{j=1}^n j^p}\right)^{\frac 1p} .
\end{equation}
For positive values of $p$, combining with lemma~\ref{lem:simplesum} yields an approximation factor of $\rho (p+1)^{\frac 1p}$.\\

We now prove the corollary for the minimum entropy set cover problem. Using a greedy algorithm for the maximum independent set, we have $\rho = (\Delta + 2)/3$~\cite{HR97}. This ratio is valid for each step of the algorithm, as the maximum degree of the graph cannot increase. From (\ref{eqn:entgeom}), the error term for the minimum entropy problem is at most
$$
\lim_{p\to 0} \log_2 \left(\frac{\Delta + 2}{3}(p+1)^{\frac 1p}\right) = \log_2(\Delta + 2) + \log_2 (e) - \log_2 (3) < \log_2(\Delta + 2) - 0.14226.
$$

\end{document}